\documentclass[a4paper,leqno]{amsart}

\usepackage{amsmath}
\usepackage{amsthm}
\usepackage{amssymb}
\usepackage{amsfonts}
\usepackage[applemac]{inputenc}
\usepackage{mathrsfs}
\usepackage{pdfsync}
\usepackage{color}
\usepackage{ulem}

\def\C{{\mathbb C}}% complex numbers
\def\R{{\mathbb R}}% real numbers
\def\N{{\mathbb N}}% nonnegative integers
\def\e{\varepsilon}
\def\eps{\varepsilon}

\def\le{\leqslant}% lessoreqal
\def\ge{\geqslant}%greaterorequal

\def\Eq#1#2{\mathop{\sim}\limits_{#1\rightarrow#2}}

\theoremstyle{plain}
\newtheorem{theorem}{Theorem}[section]
\newtheorem{lemma}[theorem]{Lemma}
\newtheorem{corollary}[theorem]{Corollary}
\newtheorem{proposition}[theorem]{Proposition}
\newtheorem{hyp}{Assumption}

\theoremstyle{definition}

\newtheorem{remark}[theorem]{Remark}
\newtheorem*{remark*}{Remark}

\numberwithin{equation}{section}

\begin{document}

\title[Weakly nonlinear time-adiabatic theory]
{Weakly nonlinear time-adiabatic theory}
\author[C. Sparber]{Christof Sparber}

\address[C.~Sparber]
{Department of Mathematics, Statistics, and Computer Science, M/C 249, University of Illinois at Chicago, 851 S. Morgan Street, Chicago, IL 60607, USA}
\email{sparber@math.uic.edu}

\begin{abstract}
We revisit the time-adiabatic theorem of quantum mechanics and show that it can be extended 
to weakly nonlinear situations, that is to nonlinear Schr\"odinger equations in which either the 
nonlinear coupling constant or, equivalently, the solution is asymptotically small. To this end, a 
notion of criticality is introduced at which the linear bound states stay adiabatically stable, but 
nonlinear effects start to show up at leading order in the form 
of a slowly varying nonlinear phase modulation. In addition, we prove that in the same regime 
a class of nonlinear bound states also stays adiabatically stable, at least in terms of spectral projections.
\end{abstract}

\date{\today}


\thanks{This publication is based on work supported by the NSF through grant no. DMS-1161580}
\maketitle

\section{Introduction}
\label{sec:intro}

The time-adiabatic theorem of quantum mechanics is concerned with systems governed by a {\it slowly varying} time-dependent 
(self-adjoint) Hamiltonian operator $H=H(\eps \tau)$, where $0<\eps\ll1$ is a small adiabatic parameter, controlling the time-scales 
on which $H$ varies. The associated Cauchy problem, governing the time-evolution of the quantum mechanical 
wave function $\Psi=\Psi(\tau, x)$, with $x\in \R^d$, reads
\begin{equation}\label{eq:Psi}
i  \partial_ \tau \Psi = H(\eps \tau) \Psi, \quad \Psi _{\mid \tau=\tau_0} = \Psi_{\rm in}(x).
\end{equation}
In the following, it will be more convenient to rewrite the system 
using the (slow) {\it macroscopic time variable} $t=\eps \tau$. In this case, \eqref{eq:Psi} becomes a singularly perturbed problem of the form
\begin{equation}\label{eq:linSch}
i \eps  \partial_t \Psi^\e = H( t) \Psi^\e, \quad \Psi^\e _{\mid t=t_0} = \Psi^\e_{\rm in}(x),
\end{equation}
where $\Psi^\e(t,x)\equiv \Psi(t/\eps ,x)$. A typical example for the time-dependent Hamiltonian $H(t)$, 
and the one we will be concerned with, is given by
\begin{equation}\label{eq:ham}
H(t):= -\frac{1}{2} \Delta + V(t,x), 
\end{equation}
where $V(t,x)$ describes some time-dependent (real-valued) potential. 

It is well-known that in the case where $V=V(x)$ is {\it time-independent}, 
the spectral theorem of self-adjoint operators allows for a precise description of the time-evolution associated to \eqref{eq:linSch}. 
In particular, it implies that if the initial data $\Psi^\e_{\rm in}$ is concentrated in a given spectral subspace of $H$, then it will remain so for all times. 
However, as soon as $H=H(t)$, the spectral subspaces (in general) start to mix during the 
time-evolution, and thus we do not have any precise information on the solution $\Psi(t,\cdot)$.

However, one might hope that for small $0<\e\ll 1$ there is a remedy to the situation. To this end, let us assume that the spectral
subspaces of $H(t)$ vary smoothly in time for $t\in [0,T]$, and that the initial wave function $\Psi^\e_{\rm in}$ is concentrated in 
one of these subspaces. Then the classical time-adiabatic theorem of quantum mechanics states that, for sufficiently small $\e\ll 1$, 
the solution $\Psi^\e(t, \cdot)$ approximately (i.e. up to a certain error which vanishes as $\e\to 0$) remains  
within the same subspace, provided the latter stays isolated from 
the rest of the spectrum of $H(t)$ for all $t\in [0,T]$, see below. In this situation, the spectral subspace is said to be {\it adiabatically stable} under the time-evolution. 
Note that in the unscaled variable $\tau $ this result corresponds to an approximation on time-scales of order $\tau\sim O(1/\eps)$.
The first adiabatic result for quantum systems appeared as early as 1928, cf. \cite{BoFo}. Since then, many mathematical extensions and developments have taken place, 
see, e.g., \cite{AvEl, Be1, Jo, Ka, Ne1, Ne2, Sc}, and the references therein. For a general introduction to this subject we refer to \cite{Te}.

A possible way of introducing the slow parameter $\e$ is to think about a quantum mechanical experiment in which the experimentalist is allowed to slowly tune 
the external potential $V=V(\eps \tau, x)$. With this in mind, it is worth noting that modern quantum mechanical experiments are often performed on ultra-cold quantum gases 
in the state of their {\it Bose-Einstein condensation} \cite{PiSt}. Indeed, ultra-cold quantum gases offer a superb level of control, unprecedented in several respects, which has 
triggered a vast amount of scientific activity, both theoretical and experimental. 
It is well-known that within a mean-field approximation the (macroscopic) wave function of the condensate is accurately described by a {\it nonlinear} Schr\"odinger 
(or, Gross-Pitaevskii) equation, cf. \cite{PiSt} for a general discussion, and \cite{ESY, Pi} and the references therein for a rigorous mathematical justification. 
It therefore seems a natural question to ask, whether one can extend the 
results of time-adiabatic perturbation theory to the case of nonlinear Schr\"odinger equations (NLS). This work is a first, modest attempt in this direction, although one should mention 
that there exist some non-rigorous works in the physics literature, cf. \cite{Yu}. 
Moreover, one should distinguish our time-adiabatic setting from the one in \cite{Sa}, which studies solitary wave solutions to nonlinear Schr\"odinger equations in 
a {\it space-adiabatic} situation, i.e., with a potential of the form $V=V(t,\eps x)$. In addition, we mention \cite{CF, CMS} both of which include rigorous results for related 
nonlinear adiabatic situations.

To be more concrete, we shall study the Cauchy problem corresponding to the following class of NLS:
\begin{equation}\label{eq:NLS}
i \eps \partial_t \Psi^\e = -\frac{1}{2} \Delta \Psi^\e + V(t,x) \Psi^\e + \lambda |\Psi^\e|^{2\sigma} \Psi^\e, \quad \Psi^\e_{\mid t=t_0}  = \Psi^\e_{\rm in}(x),
\end{equation}
where $\sigma \in \N$, and where $\lambda \in \R$ denotes a nonlinear coupling constant, describing either focusing or defocusing behavior, cf. \cite{SuSu} for a broad discussion of these terms. 
The cubic case $\sigma =1$ corresponds to the classical Gross-Pitaevskii equation. 
Clearly, an extension of the time-adiabatic theorem to such nonlinear models is not straightforward, in particular due to the lack of a 
spectral theory for general nonlinear operators. The basic idea in the present paper is to work in an asymptotic regime for which the nonlinearity can be considered as a small  
perturbation of the associated linear problem. A possible way to do so is to restrict ourselves to asymptotically small solutions of the form 
\begin{equation}\label{Psi}
\Psi^\e(t,x) = \e^{1/(2\sigma)} \psi^\e(t,x),
\end{equation}
where, as $\e \to 0$, we formally regard $\psi^\e\sim  O(1)$, say in $L^2(\R^d)$. Note that the size of the original wave function is then $\Psi^\e \sim O(\e^{1/(2\sigma)})$ 
and hence it becomes asymptotically larger the larger $\sigma\in \N$. Rewriting \eqref{eq:NLS} in terms of the new unknown $\psi^\e$ yields
\begin{equation}\label{eq:nls}
i \eps \partial_t \psi^\e = -\frac{1}{2} \Delta \psi^\e + V(t,x) \psi^\e + \lambda \e |\psi^\e|^{2\sigma} \psi^\e, \quad \psi^\e_{\mid t=t_0}  = \psi^\e_{\rm in}(x),
\end{equation}
with an effective nonlinear coupling constant of size $|\lambda^\e| = |\lambda| \e \ll 1$. The Cauchy problem \eqref{eq:nls} can thus be considered {\it weakly nonlinear}.

As we shall see below, a nonlinear coupling constant of order $O(\eps)$ will be critical for our analysis, since it corresponds to the threshold for which nonlinear effects are present at 
the leading order description of $\psi^\e$. In particular, if $\lambda^\e $ were even smaller, the problem would become essentially linearizable (as we will show below, cf., Remark \ref{rem:critical}). 
The main result of this work can now be stated as follows:

\begin{theorem}\label{thm:main}
Let $\sigma \in \N$, $\lambda \in \R$, $I \subseteq \R$ be an open time-interval containing $t_0\in \R$, and $V\in C_{\rm b}^1(I; \mathcal S(\R^d))$. 
Assume that there exists a simple eigenvalue $E(t)\in {\rm spec}(H(t))$ which stays separated from the 
rest of the spectrum by some $\delta>0$, i.e.
\[
\inf_{t\in I} {\rm dist} \big(E(t), {\rm spec}(H(t)) \setminus \{E(t)\}\big) = \delta ,
\]
and choose an associated normalized eigenfunction $\chi \in C^1_{\rm b}(I; H^{s}(\R^d))$ for $s\ge 0$. 
Finally, let $k>\frac{d}{2}$ and assume that at $t=t_0$, the initial data is concentrated in the eigenspace corresponding to $E(t_0)$, such that
\[
\| \psi^\e_{\rm in}  - \chi (t_0,\cdot) -\eps \gamma^\e \|_{H^k(\R^d)} \le C_0 \e^{k+1},
\]
where $\gamma^\e \in H^k(\R^d)$ is a corrector which is constructed
according to \eqref{eq:corrector} and satisfies $\langle \chi(t_0, \cdot), \gamma^\e \rangle_{L^2} =0$.

Then, for any compact time-interval $J\subset I$ containing $t_0$, there exists $\eps_0(J)<1$, and a constant $C>0$, 
such that for any $0<\e\le \e_0(J)$ the unique solution $\psi^\e \in C(J; H^k(\R^d))$ to the nonlinear Schr\"odinger equation \eqref{eq:nls} exists, and, in addition,
\[
\sup_{t\in J} \left \| \psi^\e(t,\cdot)  -  \chi(t,\cdot) e^{-i \varphi^\e(t)} \right\|_{L^2(\R^d)} \le C\e,
\]
where the phase $\varphi^\e(t)\in \R$ is given by
\[
\varphi^\e(t) = \frac{1}{\eps} \int_{t_0}^t E(s) ds +  \lambda \int_{t_0}^t \| \chi(s,\cdot)\|_{L^{2\sigma +2}}^{2\sigma+2} \,ds -i  \beta(t) ,
\]
with $\beta(t)\in i \R$ the Berry phase, defined in \eqref{eq:berry}.
\end{theorem}
We see that the presence of the nonlinearity shows up in the form of a slowly varying phase modulation within the 
leading order approximation of $\psi^\e$. An immediate consequence is the following corollary for the associated
spectral projectors (for which we use Dirac's notation):
\begin{corollary}
Under the same assumptions as in Theorem \ref{thm:main}, we obtain 
\[
\sup_{t\in J} \big \| |\psi^\e(t,\cdot)\rangle \langle  \psi^\e(t,\cdot) | -  |\chi(t,\cdot)\rangle \langle  \chi(t,\cdot)| \big\|_{L^2\to L^2} \le C\e,
\]
\end{corollary}
In other words, in terms of spectral projections, the {\it linear} time-adiabatic theorem is still valid under weakly
nonlinear perturbations of the form \eqref{eq:nls}.

The proof of Theorem \ref{thm:main} relies on a multiple scales expansion which yields an approximate solution to \eqref{eq:nls}. This approximate solution
is proven to be well-defined (under our assumptions) and, in a second step, used to infer the existence of a true 
solution $\psi^\eps$ which is asymptotically close to the approximate one for times of order $O(1)$.
Unfortunately, in order to control the nonlinear effects within our method of proof, we require that the initial data is sufficiently well-prepared (in the sense described above), 
even if one is only interested in the leading order approximation. In the language of, e.g., \cite{Te}, we require 
the initial data to be concentrated in a {\it super-adiabatic} subspace. This is very similar to the situation encountered in \cite{CMS}, where the 
semiclassical asymptotics for weakly nonlinear Schr\"odinger equations 
with highly oscillatory periodic potentials is studied. In fact, the basic strategy for the proof of Theorem \ref{thm:main} is similar to the one used in \cite{CMS}.

Clearly, Theorem \ref{thm:main} can be reformulated in terms of $\Psi^\e$, yielding a time-adiabatic result for asymptotically small solutions of order $O(\sqrt{\e})$. In 
particular,we have
\begin{equation}\label{esti}
\sup_{t\in J} \big \| |\Psi^\e(t,\cdot)\rangle \langle  \Psi^\e(t,\cdot) | -  \e |\chi(t,\cdot)\rangle \langle  \chi(t,\cdot)| \big\|_{L^2\to L^2} \le C\e^{2},
\end{equation}
where $\Psi^\e = \sqrt{\eps}\psi^\e$. 
In this case, a connection to the theory of {\it nonlinear bound states} for NLS equations becomes apparent. To this end, consider the ``stationary" Schr\"odinger 
equation associated to \eqref{eq:NLS}, i.e., 
\begin{equation}\label{statnonl}
 -\frac{1}{2} \Delta \Phi + V(t,x) \Phi + \lambda |\Phi|^{2\sigma} \Phi = E_* \Phi,
\end{equation}
where $E_*\equiv E_*(t)\in \R$ is a nonlinear energy-eigenvalue. Now, let $t\in \R$ be {\it fixed} assume that 
the potential $V(t,x)$ is such that $H(t)$ has a discrete (linear) eigenvalue/eigenfunction pair $(E, \chi)$. Then, classical bifurcation theory (cf., \cite{Ni}) 
implies that for $E_*\approx E$ a nonlinear bound state solution $\Phi$ exists, which is approximately given by a small multiple of $\chi$ (see Section \ref{sec:nonlin} for more details).
In the context of NLS, this has been rigorously proved in a number of papers, cf. \cite{PiWa, RoWe, Ts, TY, We}. 
Combining this fact with the result in Theorem \ref{thm:main} 
will allow us to prove that in the regime of $O(\sqrt{\eps})$-solutions these nonlinear bound states are also adiabatically stable, at least in terms of spectral projections. 
More precisely, we have the following result.

\begin{corollary} \label{cor:nonlin} 
In addition to the assumptions of Theorem \ref{thm:main}, assume that for a compact $J\subset I$ there is an $M_0(J)>0$ such that for all $M\le M_0$, 
there exists a family $\Phi\in C(J, H^2(\R^d))$ of solutions to \eqref{statnonl} with $\| \Phi(t,\cdot)\|_{L^2} = M$, for all $t\in J$. 
Let $\eps \in (0,\min(\eps_0(J), M_0(J))]$ and denote by $\Phi^\e(t,x)$ a family of bound states with norm $\| \Phi ^\e (t,\cdot)\|_{L^2} = \sqrt{\eps}$. 
Furthermore, let $\Psi^\e$ be the solution to \eqref{eq:NLS} with initial data $\Psi^\e_{\rm in}= \sqrt{\e} \psi_{\rm in}^\e$, 
where $\psi^\e$ is as in Theorem \ref{thm:main}.

Then there exist a constant $K>0$ such that
\[
\sup_{t\in J} \big \| |\Psi^\e(t,\cdot)\rangle \langle  \Psi^\e(t,\cdot) | -  |\Phi^\e(t,\cdot)\rangle \langle  \Phi^\e(t,\cdot)| \big\|_{L^2\to L^2}  \le K \eps^{\min(2, \sigma+1/2)}.
\]
\end{corollary}
Note that for $\sigma=1$ (cubic nonlinearity) this estimate is not as good as the one in \eqref{esti}. In addition, Corollary \ref{cor:nonlin} does not provide any information 
on the phase, an issue which will be discussed in more detail in Section \ref{sec:nonlin}.

We finally want to mention that in the recent PhD thesis \cite{Gr} (see also the upcoming paper \cite{GG}) a similar result has been proven (for $d=3$, $\lambda>0$ and $\sigma =1$). 
Namely, that there exists a family of nonlinear bound states $t\mapsto \Phi(t,\cdot)$ with $\| \Phi (t,\cdot) \|_{L^2} = M\ll 1$ sufficiently small, such that 
for $\eps \ll 1$ the solution  $\Psi^\e$ of \eqref{eq:NLS} satisfies
\[
\sup_{t\in [0,1]} \big \| |\Psi^\e(t,\cdot)\rangle \langle  \Psi^\e(t,\cdot) | -  |\Phi(t,\cdot)\rangle \langle  \Phi(t,\cdot)| \big\|_{L^2\to L^2}  \lesssim  \eps.
\]
Indeed, the author of \cite{Gr} proves a slightly stronger statement in terms of the wave function, not only the spectral projections (their phase information, however, is not 
as precise as the one we have in Theorem \ref{thm:main}). Note that Corollary \ref{cor:nonlin} is consistent with the one from \cite{Gr}, 
if one specifies the latter to solutions of size $O(\sqrt{\eps})$ (but we require more assumptions on the initial data due to the fact that we need to go through the proof of Theorem \ref{thm:main}). 
The main difference between the two results seems to be that in our case there is only a single small parameter $\eps$ (induced by the 
dimensionless form of the NLS itself) in which the size of the solution (or, equivalently, the strength of the nonlinearity) is measured. In comparison to that, \cite{Gr} considers 
both $M$ and $\eps$ to be small and independent, so that for a given $M\ll 1$ one can choose an $\eps$ sufficiently small 
to derive a nonlinear time-adiabatic result.\\

The paper is now organized as follows: In Section \ref{sec:constr} we shall show how to obtain the leading order approximation by means of 
formal asymptotic expansions. These expansions will then be made mathematically rigorous in Section \ref{sec:frame}. The nonlinear stability 
of our approximation is proved in Section \ref{sec:stab}, yielding the proof of Theorem \ref{thm:main}. Possible extensions and variations of 
our results, in particular, the proof 
of Corollary \ref{cor:nonlin} are then discussed in Section \ref{sec:ext}.

\section{Formal construction of the approximate solution}\label{sec:constr}

In this section we shall first recall the standard adiabatic expansion for linear equations of Schr\"odinger type and show (in a second step) how to 
include the case of (sub-)critical nonlinearities. 
To this end, we consider \eqref{eq:nls} in the case $\lambda =0$ and seek a solution in the form 
\[
\psi^\e(t,x) = e^{-i \varphi(t) / \e} \mathcal U^\e(t,x),
\]
where $\varphi(t)\in \R$ is some sufficiently smooth phase function, and the complex-valued amplitude $\mathcal U^\e$ is assumed to be of the form
\begin{equation}\label{eq:U}
\mathcal U^\e(t,x) \Eq \eps 0 \sum_{n \ge 0} \e^n U_n(t,x),
\end{equation}
in the sense of formal asymptotic expansions. Plugging this into \eqref{eq:nls} yields
\[
H(t) \mathcal U^\e = \dot \varphi(t) \mathcal U^\e + i \e \partial_t \mathcal U^\e ,
\]
where $H(t)$ is given by \eqref{eq:ham}.
Next, we plug in \eqref{eq:U} and equate powers in $\e$.
At leading order, i.e., by equating terms of order $O(1)$, 
we find:
\begin{equation}\label{eq:EVproblem}
 H(t) U_0(t,x) = \dot \varphi(t) U_0 (t,x) .
\end{equation}
This can be seen as an eigenvalue problem for the operator $H(t)$ with eigenvalue $E(t)=\dot \varphi(t)$ and we consequently conclude that
\begin{equation}\label{eq:phi}
\varphi(t) = \int_{t_0}^t E(s) \, ds,
\end{equation}
the so-called dynamic phase. Let $I\subseteq \R$ be an open time-interval containing $t_0$. 
Assuming for the moment that $E(t)$ is a simple eigenvalue for all $t\in I$, with associated normalized eigenfunction $\chi(t,\cdot)\in L^2(\R^d)$, 
we infer $U_0(t,x) = u_0(t) \chi(t,x)$, for some yet to be determined coefficient function $u_0(t)\in \C$. 

Next, by equating terms of order $O(\e)$, we find the following inhomogeneous equation
\begin{equation}\label{eq:step21}
H(t) U_1(t,x) = \dot \varphi(t) U_1(t,x) + i \partial_t U_0(t,x).
\end{equation}
Using the information from the step before, this can be rewritten as
\begin{equation}\label{eq:step2}
L_E(t) U_1(t,x) = i\big( \dot u_0(t) \chi(t,x)+ u_0(t) \partial_t \chi(t,x) \big) .
\end{equation}
where from now on, we shall denote 
\[
L_E(t) = H(t) - E(t).
\]
The kernel of $L_E(t)$ is given by ${\rm span}(\chi(t,\cdot))$ and we consequently decompose
\begin{equation}\label{eq:U1}
U_1(t,x) =  u_1(t) \chi(t,x) + v_1(t,x),
\end{equation}
where $\langle v_1(t,\cdot), \chi(t,\cdot)\rangle_{L^2}=0$, for all $t\in I$.

In order to guarantee that \eqref{eq:step2} has a solution, Fredholm's alternative asserts that the right hand side of \eqref{eq:step2} 
has to be orthogonal to $\chi(t,\cdot)$, for all $t\in I$. Taking the $L^2(\R^d)$ inner product of \eqref{eq:step2}
with $\chi$ gives
\begin{equation}\label{eq:ode}
\frac{d u_0} {dt} + u_0 \langle \partial_t \chi(t,\cdot), \chi(t,\cdot)\rangle_{L^2} = 0,
\end{equation}
and thus (up to a multiplicative constant which we shall choose equal to 1 for simplicity), we find
\[
u_0(t) = e^{-\beta(t)},
\]
where 
\begin{equation}\label{eq:berry}
\beta(t) =\int_{t_0}^t  \langle \partial_t \chi(s,\cdot), \chi(s,\cdot) \rangle_{L^2} \, ds,
\end{equation}
denotes the Berry phase term \cite{Be1, Be}. 
Note that $\beta(t) \in i \R$, for all $t\in I$, as one can easily see from differentiating the normalization condition 
$\langle \chi(t, \cdot), \chi(t,\cdot)\rangle_{L^2} = 1$. As a consequence, we also have that solutions to \eqref{eq:ode} satisfy \[
|u(t,x)| = |u(t_0,x)|, \quad \forall t\in I.\]
In summary, we find the well known leading order approximation of linear 
time-adiabatic theory. Namely, that for $\e \to 0$ the solution $\psi^\e$ behaves like
\[
\psi^\e(t,x)\Eq \eps 0  \exp\left(-  \frac{i}{\e} \int_{t_0}^t E(s) \, ds - \beta(t) \right)\chi(t,x).
\]
\begin{remark} \label{rem:berry} Concerning the significance of the Berry phase, we first note that the eigenvalue equation 
\[
H(t) \chi(t,\cdot) = E(t) \chi(t,\cdot),
\]
does not uniquely determine the eigenfunction $\chi(t,\cdot)$, even if one imposes the normalization condition $\| \chi(t,\cdot) \|_{L^2}=1$. 
One still has the freedom to change
\[
\chi(t,\cdot) \to \tilde \chi(t, \cdot) = \chi(t, \cdot) e^{i S(t)}
\]
with some (smooth) phase $S(t)\in \R$. Under such a gauge transformation the Berry phase changes by
\[
\beta(t) \to \tilde \beta(t) = \beta(t) + iS(t) - iS(t_0).
\]
This allows to choose an appropriate $S$ such that $\tilde \beta \equiv 0$, which is the reason why the Berry phase has historically been ignored for quite some time. 
However, in the case of a periodic Hamiltonian, i.e., $H(t_0) = H(t_0+T)$ for some $T>0$, one usually imposes (for physical reasons) 
the additional assumption that $\chi(t_0,\cdot) = \chi(t_0+T, \cdot)$.
In this case, the fact that $\tilde \chi$ should be single valued implies
\[
S(t_0) - S(t_0+T) = 2\pi n, \quad n \in \mathbb Z.
\]
This shows that $\beta(t_0 +T)$ can only be changed by an integer multiple of $2\pi i$ and thus can not be gauged away in general, as has been noted in \cite{Be1, Be}. 
Note that this problem remains, even in the case where $H(t)$ is real (as it is true in the present work) and thus $H(t)$ commutes with the operator of complex conjugation $\Theta$. 
In this case we infer that $[P(t), \Theta]=0$, and hence there exists a real-valued eigenfunction $\chi(t, \cdot) \in {\rm ran} \, P(t)$ which consequently 
satisfies $\langle \partial_t \chi(t,\cdot) , \chi (t,\cdot) \rangle_{L^2} \equiv 0$.
This eigenfunction, however, does not necessarily satisfy the periodicity condition imposed before. 
In the present work, we do not want to exclude the possibility of a periodic $H(t)$ and thus refrain from gauging away the Berry phase. 
For a general discussion of the physical significance of geometric phases (of which the Berry phase is one particular example), we refer to \cite{BMKNZ, XMN}. 
\end{remark}

With this in hand, it is possible to determine $v_1$ through \eqref{eq:step2}. At least formally, this yields
\[
v_1(t,x) = i L^{-1}_E(t) \Big( \dot u_0(t) \chi (t,x)+ u_0(t) \partial_t \chi(t,x) \Big) =  i L^{-1}_E(t) \Big(u_0(t) \partial_t \chi(t,x) \Big),
\]
where we denote the partial inverse (or, partial resolvent) of $L_E(t)$ by
\begin{equation}\label{eq:resol}
L^{-1}_E(t)  : =  (1-P(t)) (H(t)-E(t))^{-1} (1-P(t)),
\end{equation}
with $P (t)=| \chi(t,\cdot) \rangle \langle \chi(t,\cdot) |$ being the projection onto the eigenspace corresponding to $E(t)\in \R$.
Note that this also shows that initially $v_1(t_0,x)   \not =0$, in general. 

The remaining unknown $u_1$ appearing in \eqref{eq:U1} can then be obtained by equating terms of order $O(\e^2)$. 
Indeed, by looking at the solvability condition for 
\[
L_E(t) U_2(t,x) = i\partial_t U_1\equiv  i\left( \dot u_1(t) \chi + u_1(t) \partial_t \chi(t,x) \right) + i \partial_t v_1,
\]
one finds that $u_1(t)$ solves the following differential equation
\[
\dot u_1  + \beta(t) u_1 + \langle \partial_t v_1(t,\cdot), \chi (t,\cdot) \rangle_{L^2} =0.
\]
Choosing, for simplicity, $u_1(t_0)=0$, we get
\[
u_1(t) = - e^{-\beta(t)} \int_{t_0}^t \langle \partial_t v_1(s,\cdot), \chi (s,\cdot)\rangle_{L^2} \, e^{\beta(s)} \, ds.
\]
By repeating these steps, one easily finds that all amplitudes $U_n(t,x)$, $n\ge 1$, appearing in \eqref{eq:U}, are of the form 
\begin{equation}\label{eq:Un}
U_n(t,x) = u_n(t)\chi(t,x)+ v_n (t,x) ,
\end{equation}
where every $u_n(t)$ is determined through an ordinary differential equation obtained from the solvability condition at order 
$O(\e^{n+1})$, together with the initial data $u_n(t_0)=0$.\\

Next, we want to understand how to take into account a (sub-)critical nonlinearity in our asymptotic expansion. 
To this end, we first note that \eqref{eq:U} yields
\[
\e |\mathcal U^\e|^{2\sigma}\mathcal U^\e \sim \e |U_0|^{2\sigma}U_0 + \e^{2}\left( (\sigma +1) |U_0|^{2\sigma} U_1 +  \e \sigma |U_0|^{2\sigma -2} U_0^2 \overline U_1\right)  + O(\eps^{3}).
\]
Thus, the leading order eigenvalue problem \eqref{eq:EVproblem} does not change. The nonlinearity enters only in the expressions of order $O(\e)$ or higher. For the former 
we find the following analog of \eqref{eq:step21}:
\begin{equation}\label{eq:nlstep2}
L_E(t) U_1(t,x) = i \partial_t U_0(t,x)-  \lambda |U_0|^{2 \sigma} U_0.
\end{equation}
Here, we can use our knowledge from before to make the following ansatz for $U_0$:
\begin{equation}\label{eq:nlamp}
U_0(t,x) = \chi(t,x) e^{ - \beta(t) - i \theta(t)},
\end{equation}
where $\beta(t)$ is defined in \eqref{eq:berry} and $\theta(t)\in\R$ is some other phase yet to be determined. 
By doing so, the solvability condition requiring that the right hand side of \eqref{eq:nlstep2} has to be orthogonal to ${\rm ker}\, L_E(t)$ yields
\[
\frac{d \theta }{dt}= \lambda \langle |\chi(t,\cdot)|^{2\sigma} \chi(t,\cdot), \chi(t,\cdot)\rangle_{L^2} =\lambda  \int_{\R^d} |\chi(t,x)|^{2\sigma+2} \, dx,
\]
where we have used the fact that $\beta(t) \in i\R$. Assuming, for the moment, that $\chi(t,\cdot) \in L^{2\sigma +2}(\R^d)$, we thus find 
\begin{equation}\label{eq:nlphase}
\theta(t) = \lambda \int_{t_0}^t \| \chi(s,\cdot)\|_{L^{2\sigma +2}}^{2\sigma+2} \, ds.
\end{equation}
In view of \eqref{eq:nlamp}, we see that the nonlinearity contributes at leading order by 
adding an additional nonlinear phase modulation, i.e.,
\[
\psi^\e(t,x)\Eq \eps 0  \exp\left(- \frac{i}{\eps} \int_{t_0}^t E(s) \, ds -i \lambda \int_{t_0}^t \| \chi(s,\cdot)\|_{L^{2\sigma +2}}^{2\sigma+2} \,ds  - \beta(t)  \right)\chi(t,x).
\]
Note that even though the nonlinear phase modulation is slowly varying, one should still think of it as a small (i.e., of order $O(\eps)$) nonlinear modification of the dynamical phase. 

\begin{remark} \label{rem:critical}
It is clear by now that the choice \eqref{Psi} is critical with respect to our asymptotic expansion. Indeed, if instead of \eqref{Psi} we set 
\[
\Psi^\e(t,x) = \e^{\alpha/(2\sigma)} \psi^\e(t,x),
\]
then instead of \eqref{eq:nls} we would obtain
\begin{equation}\label{eq:nlsalpha}
i \eps \partial_t \psi^\e = -\frac{1}{2} \Delta \psi^\e + V(t,x) \psi^\e + \lambda \e^\alpha |\psi^\e|^{2\sigma} \psi^\e
\end{equation}
Performing the same asymptotic expansion as before, we see that if $ \alpha \ge 2$, then {\it no} nonlinear effects are present in the leading order asymptotics. 
The problem thus becomes essentially linearizable, 
and can be considered sub-critical with respect to our asymptotic analysis. A somewhat intermediate regime is obtained in 
the case where $\alpha$ is no longer a natural number 
and such that $1<\alpha  <2$. This situation will be discussed in more detail in Section \ref{sec:inter}.
Finally, if $0\le \alpha <1$, the problem can be considered super-critical with respect to our asymptotic expansion. 
The case $\alpha =0$ is probably the most relevant from the physics point of view, but clearly also mathematically much more challenging and 
thus beyond the scope of the present work. One can expect this problem to be intimately 
related to the modulation stability of nonlinear ground states studied in \cite{We} (see also \cite{SoWe}).
\end{remark}

%%%%%%%%%%%%%%%%%%%%%%%%%%%%

\section{A mathematical framework for asymptotic expansions}\label{sec:frame}

In this section we will prove that the solution obtained through the formal multiple scales approximation above is indeed well defined and furnishes 
an approximate solution to \eqref{eq:nls}. To this end, we shall impose the following basic assumptions on
the time-dependent potential:

\begin{hyp}\label{ass1} The potential $V(t,x)$ satisfies $V\in C^1_{\rm b}(I; \mathcal S(\R^d))$, where $\mathcal S$ denotes the space of smooth 
and rapidly decaying functions.
%Let $k\in \N$, and the potential $V(t,x)$ be such that $V\in C_{\rm b}^1(I; C_{\rm b}^{k}(\R^d))$, and 
%\[
%\lim_{|x|\to \infty} |V(t,x)|=0, \quad \forall \, t\in I.
%\]
\end{hyp}
\begin{remark}
This assumption is mainly imposed for the sake of a simple and clean presentation but 
certainly far from optimal concerning the regularity of $V$. Indeed, all of our results can be reformulated for 
potentials $V(t,\cdot)\in C_{\rm b}^{k}(\R^d)$ and vanishing at infinity. However, it turns out that 
at different stages of our proofs we require different (and relatively strong) bounds on $k$, which 
result in a somewhat tedious regularity count that we want to avoid. 
\end{remark}

Nest, we fix $t\in I$. Then it is well known (see, e.g., \cite[Chapter 10.1]{Tebook}) that for $V(t,\cdot) $ bounded, i.e., $k=0$, and decaying at infinity,
the Hamiltonian $H(t)$ is a self-adjoint operator with $\text{dom}(H(t))=H^2(\R^d)\subset L^2(\R^d)$. 
Moreover, for any fixed $t\in I$ the spectrum of $H(t)$ is of the standard form, i.e.,
\[
\text{spec}(H(t)) = [0,\infty) \cup \{ E_j(t) \, | \, - E_j(t)>0, \quad j=0,1, \dots \}.
\]
see, e.g., \cite{Hi}.
Of course as these eigenvalues vary in time, they might cross each other, or disappear into the continuous spectrum. Our main assumption is 
that the eigenvalue $E(t)$ we are interested in stays separated from the rest of the spectrum by a spectral gap. 

\begin{hyp}\label{ass2}
There exists a simple eigenvalue $E(t)\in {\rm spec}(H(t))$ and a constant $\delta>0$, satisfying
\begin{equation}\label{gap}
\inf_{t\in I} {\rm dist} (E(t), {\rm spec}\big(H(t) \setminus E(t))\big) = \delta .
\end{equation}
\end{hyp}
Note that this implies $E(t) \le -\delta$, for all $t\in I$. 
Denoting by $\chi(t, \cdot)\in L^2(\R^d)$ a normalized eigenfunction corresponding to such a well separated eigenvalue $E(t)$, we have the following regularity result.
\begin{lemma}\label{chiregularity}
Let Assumptions \ref{ass1} and \ref{ass2} hold, then we can choose $\chi: I \to H^{k}(\R^d)$, such that $\chi \in C_{\rm b}^1(I, H^{k}(\R^d))$ for any $k\ge 0$, and such that 
$\| \chi (t,\cdot) \|_{L^2}=1$, for all $t\in I$. 
\end{lemma}
\begin{proof}
The proof follows from standard arguments. Indeed, we first notice that, for any fixed $t\in \R$, 
$\chi(t, \cdot)$ satisfies the Schr\"odinger eigenvalue problem 
\[
\left(-\frac{1}{2} \Delta + V(t,x) \right)\chi(t,x) = E(t) \chi(t,x),
\]
which, in view of Assumption \ref{ass1} and \cite[Proposition 1.2]{Hi}, implies the asserted regularity in $H^{k}(\R^d)$ for any $k\ge 0$.
Thus it only remains to prove the differentiability property in time. This follows from the fact that as long as 
$E(t)$ stays separated from the rest of the spectrum, the associated orthogonal projector $P (t)$ can be expressed via Riesz' formula as 
\[
P(t) = \frac{i}{2 \pi} \oint_{\Gamma(t)} (H(t) - z)^{-1} \, dz,
\]
where $\Gamma(t)\subset \C$ is a continuous (positively oriented) curve encircling $E(t)$ once, such that 
\[
\inf_{t\in I} {\rm dist} (E(t), {\rm spec}\big(H(t))\big) = \delta/{2},
\]
i.e., no other points within $\text{spec}(H(t))$ are enclosed by $\Gamma(t)$. Using this, we see that 
\[
\frac{d}{dt} P(t) =  \frac{i}{2  \pi} \oint_{\Gamma(t)} \frac{d}{dt} (H(t) - z)^{-1}   \, dz,
\]
whenever $(H(t) - z)^{-1} \in C^1_{\rm b}(I; \mathcal L(L^2(\R^d)))$. The latter is proved for example in \cite[Lemma 2.4]{Te}. Hence $t\mapsto P(t)$ is in $C_{\rm b}^1(I)$. 

In order to obtain a corresponding eigenfunction $\chi \in C_{\rm b}^1(I; H^{k}(\R^d))$, we can follow the classical idea of Kato \cite{Ka} (see also \cite{Ne1, Ne2}), which starts from the definition of the following operators
\[
K(t):= i [P(t), \dot P(t)], 
\]
and $A(t)$ given by
\[
 \frac{d}{dt} A(t) = i K(t) A(t), \quad A(t_0) = \mathbb I.
\]
Then it is easily seen that $K(t)$ is self-adjoint and thus $A^*(t) = A^{-1}(t)$. In addition, one checks (after some calculations invoking the properties of projections)
\[
\frac{d}{dt} \left(A^*(t) P(t) A(t) \right)= 0,
\]
so that the intertwining property holds:
\begin{equation}\label{inter}
P(t) = A(t) P(t_0) A^*(t), \quad \forall t \in I.
\end{equation}
We now choose $\chi(t_0)$ normalized such that $P(t_0) \chi(t_0, \cdot) = \chi(t_0, \cdot)$ and define $\chi(t,\cdot) = A(t) \chi(t_0, \cdot)$. In view of \eqref{inter} this implies that 
$P(t) \chi(t, \cdot) = \chi(t, \cdot)$ for all $t\in I$. In addition $\| \chi(t,\cdot)\|_{L^2}=1$ and since $A\in C^1_{\rm b}(I; \mathcal L(L^2(\R^d)))$ we also get $\chi(t,\cdot) \in C^1_{\rm b}(I; H^{k}(\R^d))$.
\end{proof}

By Sobolev imbedding we also have $H^{k}(\R^d) \hookrightarrow L^{\infty}(\R^d)$, provided $k> \frac{d}{2}$. In particular, 
%for 
%\[
%s=k+2>\frac{d}{2} \Leftrightarrow k>\frac{d-4}{2},
%\] 
we have $\chi\in L^q(\R^d)$ for any $q\in [2,\infty]$ and thus the expression for the nonlinear phase modulation $\theta(t)$ given by \eqref{eq:nlphase} is well-defined.
Moreover, for $k> \frac{d}{2}$, the Sobolev space $H^{k}(\R^d)$ is in fact an algebra, i.e., 
if $f,g\in H^{k}(\R^d)$ then $fg\in H^{k}(\R^d)$. This can be used to prove the following regularity result:

\begin{lemma}\label{lem:regular}
Let $\sigma \in \N$, $\lambda \in \R$, and Assumptions \ref{ass1} and \ref{ass2} hold. %for some $k>\frac{d-4}{2}$.  
Then the expressions appearing in the asymptotic expansion \eqref{eq:U} satisfy 
$U_n\in C^1_{\rm b}(I; H^{k}(\R^d))$ for all $n \in \N$ and $k\ge 0$.
\end{lemma}

\begin{proof}
Each $U_n$ is of the form given in \eqref{eq:Un}, i.e., $U_n (t,x)= u_n(t)\chi(t,x) + v_n(t,x)$, with $v_0 \equiv 0$. 
In view of the asymptotic expansion above, we know that each $u_n(t)$ solves an ordinary differential equation of the form
\[
\dot u_n  + \beta(t) u_n + \langle \partial_t v_n(t,\cdot), \chi (t,\cdot) \rangle_{L^2}  = 
i \lambda \left\langle \frac{d^{n}}{ds^{n}}F\left(U_0 + \sum_{\ell=1}^{n} s^\ell U_\ell \right) \Big|_{s=0}, \chi(t, \cdot)\right\rangle_{L^2} ,
\]
where we denote the nonlinearity by $F(z)=|z|^{2\sigma} z$, which for $\sigma \in \N$ is smooth. 
Note that due to the orthogonality of $\chi$ with every $v_\ell$, the right hand side in fact only involves $u_\ell$'s. 
Moreover, we see that for $n\ge 1$, these differential equations are indeed linear. Together with the fact  
that $\frac{d}{dt}|u_0(t,x)|^2=0$ we infer that there is no restriction on the existence time of $u_n(t)$.
In view with Lemma \ref{chiregularity} we thus have $u_n\chi \in C_{\rm b}^1(I; H^{k}(\R^d))$ for all $k\ge0$.

On the other hand, we know that all $v_n$, for $n\ge 1$, are determined by inverting an elliptic equation for any fixed $t\in I$, i.e.,
\begin{equation}\label{eq:v}
v_{n}(t,x) = L_E^{-1}(t) \left(i \partial_t U_{n-1}(t,x) + \lambda \frac{d^{n-1}}{ds^{n-1}}F\left(U_0 + \sum_{\ell=1}^{n-1} s^\ell U_\ell \right) \Big|_{s=0} \right) .
\end{equation}
The fact that $H^{k}(\R^d)$ for $k>\frac{d}{2}$ forms an algebra implies that the right hand side, which is 
a sum of products of $U_\ell$'s, is in $H^{k}(\R^d)$ for all $k\ge 0$ (see also the proof of Proposition \ref{prop:asympt} below).
Since $L_{E}^{-1}(t): L^{2}(\R^d)\to \text{dom}(H(t))=H^2(\R^d)$, in view of \eqref{eq:resol}, the assertion follows by induction over $n$.
\end{proof}

With this result in hand, we set
\begin{equation}\label{psiN}
\psi^\e_{N} (t,x) := e^{- i \varphi(t)/\eps} \sum_{n=0}^N \e^n U_n(t,x),
\end{equation}
where $\varphi(t)$ is the dynamic phase given by \eqref{eq:phi} and $N\in \N$. Note that at $t=t_0$ the $U_n$ can in general not be chosen arbitrarily, since parts of it need to be determined recursively 
as given in \eqref{eq:v}. In particular, we have
\[
\psi^\e_{N} (t_0,x) = \chi(t_0,x) +\eps \gamma^\e(x)
\]
where, due to the regularity of $v_n$, the corrector $\gamma^\e\in H^{k}(\R^d)$ of Theorem \ref{thm:main} is of the form
\begin{equation}\label{eq:corrector}
\gamma^\e(x) = \sum_{n=1}^N \eps^{n-1}  v_n(t_0,x),
\end{equation}
with $v_n(t_0,x)$ as above. This definition of $\psi^\e_N$ then yields an {\it approximate solution} of the nonlinear Schr\"odinger equation \eqref{eq:nls} in the following sense:

\begin{proposition} \label{prop:asympt}
Let $\sigma \in \N$, $\lambda \in \R$, and Assumptions \ref{ass1} and \ref{ass2} hold.% for some $k>\frac{d-4}{2}$. 
Then $\psi^\e_N$ defined by \eqref{psiN} satisfies $\psi_N\in C_{\rm b}^1(I; H^{k}(\R^d))$ for all $k\ge 0$ and
\[
i \eps \partial_t \psi_N^\e + H(t) \psi_N^\e + \lambda \e |\psi_N^\e|^{2\sigma} \psi_N^\e = r_N^\e(t,x),
\]
where the remainder is bounded by
\[
\sup_{t\in I} \| r_N^\e (t, \cdot) \|_{H^{k}(\R^d))}\le C \e^{N  +1}.
\]
\end{proposition}
\begin{proof}
By plugging $\psi^\e_N$ into the nonlinear Schr\"odinger equations, the asymptotic expansion above shows that 
\[
r^\e_N(t,x) =   \e^{N+1} e^{- i \varphi(t)}\left( i \partial_t U_N(t,x) + \lambda \widetilde {r}^\e_N  (t,x)\right) ,
\]
where
\[
\widetilde {r}^\e_N= \sum_{j=N}^{(2\sigma+1)N}\e^{j-N}\sum_{\ell_1+\dots+ \ell_\sigma+m_1+\dots +m_\sigma+r=j} 
U_{\ell_1}\dots U_{\ell_\sigma}  \bar U_{m_1}\dots \bar U_{m_\sigma}U_r.
\]
In view of the regularity result established in Lemma \ref{lem:regular}, and the algebra property of $H^{k+2}(\R^d)$, for $k>\frac{d-4}{2}$, 
we directly obtain the estimate on the remainder stated above.
\end{proof}

This result, however, is not sufficient to conclude that the exact solution $\psi^\e$ will stay close to the approximate solution $\psi^\e_N$ for times of order $O(1)$. 
We shall show in the next section that this is indeed the case.

%%%%%%%%%%%%%%%%%%%%%%%%%%%

\section{Nonlinear stability of the approximation}\label{sec:stab}

\subsection{Preliminaries} 
Before we can prove stability of our asymptotic expansion, we need a basic (local in-time) 
existence result for solutions to nonlinear Schr\"odinger equations of the form \eqref{eq:nls}.

\begin{lemma}\label{lem:exist}
Let $\sigma \in \N$, $\lambda \in \R$, and $\N \ni k>\frac{d}{2}$. Moreover, let $\psi^\e_{\rm in} \in H^k(\R^d)$ and the potential satisfy $V\in C^1_{\rm b}
(I; \mathcal S(\R^d))$. 
Then there exist $T^\e_1, T^\e_2>0$, and a unique solution $\psi^\e \in C\big([t_0-T^\e_1,t_0+T^\e_2]; H^k(\R^d)\big)$ to \eqref{eq:nls}. Furthermore
\[ 
\| \psi^\e(t, \cdot) \|_{L^2(\R^d)} = \| \psi^\e_{\rm in} (t_0,\cdot)\|_{L^2(\R^d)}, \quad \forall t\in [t_0-T^\e_1,t_0+T^\e_2]\subset I. 
\]
\end{lemma}
\begin{proof} The proof is a straightforward extension of the one given in, e.g., \cite[Proposition 3.8]{Tao} for the case without potential.
We rewrite the NLS using Duhamel's principle
\[
\psi^\e(t,\cdot)  = e^{- i t \frac{\Delta}{2\eps }} \psi_{\rm in}^\e - i \int_{t_0}^t  e^{i (s-t)\frac{\Delta}{2\eps}} 
\left(\lambda |\psi^\e(s,\cdot)|^{2\sigma}+ \frac{1}{\eps} V(s,\cdot)\right)\psi^\e (s,\cdot) \, ds
=: \Xi(\psi^\e)(t).
\]
Clearly, the free Schr\"odinger group $e^{- i t \frac{\Delta}{2\eps }}$ is an isometry on $H^k(\R^d)$ for any $k\in \R$, and our assumptions on $V$ guarantee that 
there is a constant $C=C(k,d)>0$ such that
\[
\| V \psi^\e \|_{H^k} \le \sum_{|\alpha|\le k} \| \partial^\alpha V\|_{L^\infty} \, \| \psi^\e\|_{H^{k-\alpha}}\le C \, \| V \|_{C^k_{\rm b}}\, \| \psi^\e\|_{H^k}<\infty.
\]
Moreover, for $\sigma \in \N$, the nonlinearity 
$F(z)=|z|^{2\sigma}z$ is smooth which, together with the fact that $H^k(\R^d)$ for $k>\frac{d}{2}$ 
forms an algebra, allows us to estimate
\begin{equation}\label{eq:est}
\| \psi^\e(t,\cdot) \|_{H^k} \le \| \psi_{\rm in}^\e \|_{H^k} +   C^\e \int_{t_0}^t  \|\psi^\e(s,\cdot)\|^{2\sigma+1}_{H^k} +  \| \psi^\e (s,\cdot)\|_{H^k} \, ds,
\end{equation}
where $C^\e=  C(k,d,\lambda, V, \eps)>0$. 

Now denote by $X:=C(([t_0-T^\e_1,t_0+T^\e_2];
H^k(\R^d))$ for some $T^\e_1,T^\e_2>0$ to be chosen later on and $s>\frac{d}{2}$. Further, let $R>1$ be such that $\| \psi^\e_{\rm in} \|_{H^k}\le R$. 
Then, we can show that the $u\mapsto \Xi(u)$ maps the ball $B_{2R}(0)\subset X$ into itself. Indeed, the estimate \eqref{eq:est} implies
\begin{align*}
\| \Xi(u) \|_X \le &\ \| \psi_{\rm in}^\e \|_{H^k} + C^\e \max(T^\e_1,T^\e_2) \big( \|u\|^{2\sigma+1}_{X} +  \| u\|_{X} \big)\\
\le &\ R +C^\e \max(T^\e_1,T^\e_2) \big((2R)^{2\sigma+1} + 2R\big)\\
\le & \, R +  2^{2\sigma+2} C^\e \max(T^\e_1,T^\e_2) R^{2\sigma +1}. 
\end{align*}
Hence, we can choose $T_1^\e, T^\e_2\le \frac{R^{-2\sigma} }{2^{2\sigma+2} C^\e}$ and such that $[t_0-T^\e_1, t_0+T^\e_2]\subset I$. 
The same type of estimate shows that $u\mapsto \Xi(u)$ is a contraction on $B_{2R}(0)\subset X$ 
and hence there exists a unique fixed point $u = \psi^\e\in X$. The conservation of the $L^2$-norm of the solution then 
follows from the fact that $H(t)$ is self-adjoint and the nonlinearity is gauge invariant. 
\end{proof}
\begin{remark}
By carefully tracking the $\eps$-dependence of $T_{1,2}^\eps$, one finds that, in general, $T_{1,2}^\eps$ will go to zero, as $\eps \to 0$. 
However, the stability proof below actually shows that for our choice of initial data, one can find $T_{1,2}^\e > 0$ independent of $\eps$.
\end{remark}

We will also need the following Moser type lemma, proved in, e.g., \cite{Ra}.

\begin{lemma}\label{lem:moser}
Let $R>0$, $s\in \N$ and $F(z)=|z|^{2\sigma}z$, with $\sigma\in \N$. Then there exists
$K=K(R,s,\sigma)$ such that if $w$ satisfies 
\begin{equation*}
\left\| \partial^\beta w\right\|_{L^\infty(\R^d)} \le R,
\quad | \beta| \le s ,
\end{equation*}
and $\eta$ satisfies $\displaystyle \left\| \eta
\right\|_{L^\infty(\R^d)} \le 
R$, then 
\begin{equation*}
 \sum_{|\beta|\le s}\left\| \partial^\beta   \big(F(w+\eta)
- F(w)\big)\right\|_{L^2(\R^d)} 
\le K \sum_{| \beta| \le s}\left\| \partial^\beta\eta\right\|_{L^2(\R^d)} . 
\end{equation*}
\end{lemma}
In \cite{Ra}, this lemma was proved for $\e$-scaled derivatives. For our purposes, we can set $\e =1$ but note that the 
estimate above is linear in the $H^s$ norm of $\eta$, which subsequently allows the use of Grownwall's lemma (see the proof of Proposition \ref{prop:stab} below).

\subsection{Nonlinear stability} We are now in the position to prove the desired stability result for 
the asymptotic expansion obtained above.

\begin{proposition}\label{prop:stab}
Let $\sigma \in \N$, $\lambda \in \R$, $\N\ni k>\frac{d}{2}$, and Assumptions \ref{ass1} and \ref{ass2} hold. % with $\N\ni k>\frac{d}{2}$. 
Given an approximate solution $\psi_N^\e$ of the form \eqref{psiN} with $N>k$, 
we assume that, at $t=t_0$, the initial data $\psi^\e_{\rm in} \in H^k(\R^d)$ is such that
\[
\| \psi^\e_{\rm in}  - \psi^\e_{N-1} (t_0,\cdot) \|_{H^k(\R^d)} \le C_0 \e^{N}.
\]
Then, for any compact time-interval $J\subset I$ containing $t_0$, there exists an $\eps_0(J)>0$, and a constant $C>0$, 
such that for any $0<\e\le \e_0(J)$ the unique solution $\psi^\e \in C(J; H^k(\R^d))$ to \eqref{eq:nls} exists and, in addition,
\[
\sup_{t\in J} \| \psi^\e(t,\cdot)  - \psi^\e_{N-1} (t,\cdot) \|_{H^k(\R^d)} \le C \e^{N-k}.
\]
\end{proposition}

Note that this result in particular implies that the solution $\psi^\e$ to \eqref{eq:nls} cannot exhibit blow-up on any finite time-interval $J\subset I\subseteq\R$.

\begin{proof} 
Let $J = [t_0-T_1, t_0+T_2]\subset I$, for some $T_1, T_2>0$ independent of $\eps$. 
From Lemma \ref{lem:exist} we obtain the existence of a unique solution 
$\psi^\eps\in C([t_0-T_1^\e,t_0+T_2^\eps],H^k(\R^d))$ with $k>\frac{d}{2}$, to \eqref{eq:nls}. 
We denote the difference between the exact and the approximate solution by \[ 
\eta^\eps:= \psi^\eps-\psi_{N}^\eps.\] 
Proposition \ref{prop:asympt} then implies that
$\eta^\eps \in C([t_0-\tau_1^\e,t_0+\tau_2^\eps],H^k(\R^d))$, where $\tau_{j}^\eps=\min
(T_{j}^\eps,T_{j})$, with $j=1,2$. We prove that for $\eps$ sufficiently small, $\eta^\eps$ 
may be extended up to the time-interval $J\subset I$, with $\eta^\eps \in C(J,
H^k(\R^d))$. For simplicity, we shall only show the argument for the times bigger than $t_0$. A similar argument applies on the 
time interval $[t_0-T_1 , t_0]$.

Take $\eps_0>0$ so that $C_0\eps_0\leq\frac{1}{2}$, and for 
$\eps\in]0,\eps_0]$, let  
$$t^\eps := \sup \Big\{ t\ge t_0 \mid
  \sup_{t'\in[t_0,t]}\|\eta^\eps(t')\|_{H^k(\R^d)}\le 1\Big\}.$$
We already know that $t^\eps>0$ by the local existence result for
$\psi^\eps$. By possibly reducing $\eps_0>0$ even further, we shall show that $t^\eps\ge t_0+T_2$. 
The error satisfies
\begin{align}\label{eq:error}
i \e \partial_t \eta^\e = H(t) \eta^\e + \lambda \eps \left( |\psi^\e_N + \eta^\e|^{2\sigma} (\psi^\e_N + \eta^\e) - |\psi^\e_N |^{2\sigma} \psi^\e_N \right) + r_N^\e ,
\end{align}
subject to $ \eta^\e _{\mid t=t_0} = \eta^\e_{\rm in}(x)$, where $\|\eta^\e_{\rm in}\|_{H^k}=O(\eps^{N+1})$ by assumption. 

Next, we multiply \eqref{eq:error} by $\overline \eta^\e$, integrate over $\R^d$, and take the real part of the resulting expression. Since $H(t)$ is self-adjoint, this yields
\[
\partial_t \| \eta^\e \|_{L^2} \lesssim \big \| |\psi^\e_N + \eta^\e|^{2\sigma} (\psi^\e_N + \eta^\e) - |\psi^\e_N |^{2\sigma} \psi^\e_N \big \|_{L^2} + \frac{1}{\eps} \| r^\eps_N\|_{L^2}.
\]
In view of Proposition \ref{prop:asympt}, we have $\| r^\eps_N\|_{H^k} = O(\eps^{N+1})$. On the other hand, 
for $k>\frac{d}{2}$ the Gagliardo-Nirenberg inequality implies
\[
\| \eta^\e \|_{L^\infty} \lesssim  \| \eta^\e \|_{H^k} \lesssim 1 \quad \forall t\in [t_0, t^\e],
\]
and we will show, that in fact $\| \eta^\e\|_{L^\infty}$ is (asymptotically) small. To this end, we first recall that since
$\psi_N(t,\cdot) \in H^{m}(\R^d)$ for all $m\ge0$ we also have that $\psi_N(t,\cdot)\in W^{s, \infty}(\R^d)$, for all $s\le k$. 
Applying Lemma \ref{lem:moser} with $s=0$, we consequently obtain
\[
\partial_t \| \eta^\e \|_{L^2} \le K \| \eta^\e \|_{L^2} + C \eps^N ,
\]
for $t\in [t_0, t^\e]$ and, by using Grownwall's lemma, we thus find
\begin{equation}\label{eq:bound}
 \| \eta^\e \|_{L^2} \leq C_1  \eps^N,  \quad \forall t\in [t_0, t^\e].
\end{equation}
The idea is now to obtain a similar estimate for (weak) derivatives of $\eta^\e$, in order to close the argument in $H^k(\R^d)$. To this end, we first note that 
\[
i \e \partial_t (\nabla \eta^\e) = H(t) (\nabla \eta^\e) + [\nabla , H(t)] \eta^\e + 
\lambda \eps \nabla \left(F(\psi_N^\e+\eta^\e) - F(\psi_N^\e)\right) + 
\nabla r_N^\e ,
\]
and the same type of argument as before, together with the Cauchy Schwartz inequality, yields
\[
\partial_t \| \nabla \eta^\e \|_{L^2} \lesssim 
\|  \nabla \left(F(\psi_N^\e+\eta^\e) - F(\psi_N^\e)\right) \|_{L^2}  +\frac{1}{\eps} \| [\nabla , H(t)]  \eta^\e\|_{L^2}+ \frac{1}{\eps} \| \nabla r^\eps_N\|_{L^2}.
\]
Now $[\nabla , H(t)] = \nabla V(t,x)$, which is bounded by assumption, and so
\[
\partial_t \| \nabla \eta_N^\e \|_{L^2} \lesssim \|  \nabla \left(F(\psi_N^\e+\eta^\e) - F(\psi_N^\e)\right) \|_{L^2} + \frac{1}{\eps} \| \eta^\e \|_{L^2} + C \eps^N.
\]
Invoking again Lemma \ref{lem:moser} with $s=1$, and the bound \eqref{eq:bound}, we infer that $\forall t\in [t_0, t^\e]$ it holds
\[
\partial_t \| \nabla \eta^\e \|_{L^2} \lesssim \|  \nabla \eta^\e \|_{L^2} +  \eps^{N-1},
\]
and Grownwall's lemma, together with \eqref{eq:bound}, then yields
\[
\|  \eta^\e \|_{H^1} \lesssim  \eps^{N-1},  \quad \forall t\in [t_0, t^\e],
\]
Iterating in $s \le k$, we obtain, more generally
\begin{equation}\label{eq:kest}
\|  \eta^\e \|_{H^k} \lesssim  \eps^{N-k},  \quad \forall t\in [t_0, t^\e].
\end{equation}
and the Gagliardo-Nirenberg inequality consquently implies
\[
\| \eta^\e \|_{L^\infty} \lesssim  \| \eta^\e \|_{H^k} \lesssim \eps^{N-k} \quad \forall t\in [t_0, t^\e],
\]
provided $k>\frac{d}{2}$. For $N-k>0$, continuity of $\| \eta^\e(\cdot ,t)\|_{H^k}$ implies that $t^\e\ge t_0+T_2$, for $\eps\le \eps_0(T_2)$ sufficiently small, since if 
$t^\e < t_0 +T_2$ the estimate 
\[
\sup_{t\in [0,t^\e] }\| \eta^\e \|_{L^\infty} \lesssim  \| \eta^\e \|_{H^k} \lesssim \eps^{N-k} < \frac{1}{2},
\]
(after possibly reducing $\eps_0$) shows that $\eta^\e$ can be continued beyond $t^\e$, a contradiction. 
In particular, we obtain that  $\eta^\e$, and hence $\psi^\e$, is well defined for all $t\in[t_0, t_0 + T_2]$, thus showing $T_2^\e \ge T_2$. Since the same argument can be 
applied for times smaller than $t_0$, we finally conclude that $\psi^\e$ 
is well defined for all $t\in J = [t_0-T_1, t_0+T_2]\subset I$ and $0 < \eps \le \eps_0(J)$.

To complete the proof of the theorem, we note that \eqref{eq:kest} implies
\[
\sup_{t\in J}\|  \psi^\e - \psi^\e_{N} \|_{H^{k}} \lesssim  \eps^{N-k},
\]
and since $N>k$, we also have
\[
\sup_{t\in J}\| \psi^\e_{N} - \psi^\e_{N-1} \|_{H^{k}} \lesssim \eps^{N}=o(\eps^{N-k}).
\]
Thus, we can use the triangle inequality to replace $\psi^\e_{N}$ with $\psi^\e_{N-1}$ in our estimate, which yields the desired result.
\end{proof}

Proposition \ref{prop:stab} directly implies the result stated in Theorem \ref{thm:main}. Due to our method of proof, 
Proposition \ref{prop:stab} yields a loss in accuracy for the obtained error estimates, which we are unable to overcome at this point. 

%%%%%%%%%%%%%%%%%%%%%%%%%%%%%%%%%%%%%%%%

\section{Possible extensions and variations}\label{sec:ext}

\subsection{Remarks on closely related cases} In this section we collect several remarks on how to extend Theorem \ref{thm:main} to other, closely related, situations.

\subsubsection{Degenerate Eigenvalues} The results above readily generalize to the case of an $M$-fold degenerate eigenvalue $E(t)$, satisfying the gap condition \eqref{gap}. 
In this case there exists 
a smooth basis $\chi_{\ell}(t,\cdot)\in L^2(\R^d)$, where $\ell= 1, \dots, L$, of the associated eigenspace. The associated projection onto the eigenspace corresponding to $E(t)$ then becomes 
\[
P(t) = \sum_{\ell=1}^L | \chi_\ell(t,\cdot)\rangle \langle \chi_\ell(t,\cdot) |.
\]
Using this, one can proceed along the same lines as above to obtain that 
\[
\sup_{t \in \tilde I } \Big \|\psi^\e(t,\cdot) - e^{i \varphi(t)/\eps} \sum_{\ell=1}^L u_{0,\ell}(t) \chi_\ell(t,\cdot) \Big \|_{H^k} \le C \eps. 
\]
However, the formulas in general become more complicated, since 
the coefficient functions, $u_{\ell,0}(t)$ are now determined by an $L\times L$ system of ordinary differential equations (leading to matrix-valued 
Berry phases and analogous nonlinear phase modulations). This consequently leads to rather tedious computations in the 
subsequent steps of our asymptotic expansion, which is why we shall not go into further details here. In addition, the global (on $I$) existence 
of the solutions to the $L\times L$ system is not clear, a-priori. However, Proposition \ref{prop:stab} applies on any interval $\tilde I\subseteq I$ on which the $u_{\ell}$'s are defined.

\subsubsection{Quadratic potentials} In view of a possible application to Bose-Einstein condensates, the assumption that $V(t,x)$ vanishes as $|x|\to \infty$ seems unrealistic, since 
one typically considers trapping potentials of the form
\[
V(t,x) = \sum_{j=1}^d \Omega_j(t) x_j^2,\quad \Omega(t)\in \R,
\]
i.e., a time-dependent harmonic oscillator. There is, however, no fundamental difficulty in extending our result to such a situation. 
Indeed, as long as $\Omega_j(t)>0$, the existence of eigenvalues $E(t)$ together with their associated smooth (and rapidly decaying) eigenfunctions is guaranteed (see, e.g., \cite{Tebook}), and 
the asymptotic expansion, stays (at least formally) exactly the same as before. 
Only from the point of view of rigorous estimates, one needs to shift from the usual Sobolev space setting $H^k(\R^d)$, to weighted spaces of the form
\[
\Sigma^k = H^k(\R^d)\cap \{ |x|^k f \in L^2(\R^d) \}.
\]
The basic existence and well-posedness theory for NLS in such weighted spaces has been established in \cite{Car}, yielding a unique solution $\psi^\e \in C(\tilde I; \Sigma^k)$ on some $
\tilde I\subseteq I$, 
provided 
$\psi_{\rm in} \in \Sigma^k$. Moreover, an extension of the Schauder Lemma \ref{lem:moser} to $\Sigma^k$ is straightforward. 
The only extra work needed is in the proof of the nonlinear stability, where now $[\nabla , H(t)] = \nabla V(t,x)$ is no longer bounded. 
However, since $\| \nabla V \eta^\e \|_{L^2} \simeq \| x \eta^\e \|_{L^2}$, and since $[x , H(t)] = \nabla$, a closed set of estimates for the {\it combined} $L^2$-norms of $x\eta^\eps$ and $\nabla \eta^\eps$ 
(and thus for the $\Sigma^1$-norm of $\eta^\e$) can be obtained, cf. \cite{CMS} for more details. Iterating this then yields a stability result in $\Sigma^k$.

\subsubsection{The intermediate regime $1<\alpha <2$}\label{sec:inter}
We go back to the discussion started in Remark \ref{rem:critical} and consider the slightly more general situation of
\[
i \eps \partial_t \psi^\e = -\frac{1}{2} \Delta \psi^\e + V(t,x) \psi^\e + \lambda \e^\alpha |\psi^\e|^{2\sigma} \psi^\e, \quad \alpha \ge 1.
\]
We already know that if $\alpha =1$ the problem is critical, and that if $\alpha \ge 2$, the problem is sub-critical (i.e., linearizable). 
The intermediate regime $1<\alpha <2$, however, is slightly more complicated, since the asymptotic expansion used before fails to match the 
size of the nonlinearity. One way to overcome this problem is to include the nonlinearity in the equation of order $O(\e)$, which yields
\begin{equation*}\label{eq:gennlstep2}
L_E(t) U_1(t,x) = i \partial_t U_0(t,x)+  \e^{\alpha -1} \lambda |U_0|^{2 \sigma} U_0,
\end{equation*}
instead of \eqref{eq:nlstep2}. For $\eps^{\alpha -1}\ll 1$ this can be seen as a regular perturbation problem of the associated linear situation. The 
corresponding solvability condition now yields an $\e$-dependent leading order amplitude  of the form 
\begin{equation*}
U^\eps_0(t,x) = \chi(t,x) e^{ - \beta(t) + i \eps^{\alpha-1} \theta(t)},
\end{equation*}
with $\theta(t)$ given by \eqref{eq:nlamp}. The nonlinear phase modulation appearing in this expression is obviously 
rather weak, due to the small $\eps^{\alpha -1}\ll 1$ factor in front. The price to pay is that now all the $U_n^\e$ become $\eps$-dependent. An alternative 
approach would be to consider a modified asymptotic expansion, which includes powers of $\eps^\alpha$. This can be done in principle, but is rather cumbersome 
and we will leave the details to the reader.

\subsection{Connection to nonlinear bound states} \label{sec:nonlin}

We finally turn to the connection with nonlinear bound states and thus, to the proof of Corollary \ref{cor:nonlin}. To this end, we first recall that, 
for fixed $t\in \R$, nonlinear bound states are solutions to the stationary Schr\"odinger equation
\begin{equation}\label{eq:stat}
 -\frac{1}{2} \Delta \Phi + V(t,x) \Phi + \lambda |\Phi|^{2\sigma} \Phi = E_* \Phi.
\end{equation}
Denoting by $E(t)<0$ a simple eigenvalue of the linear Hamiltonian $H(t)$, 
standard bifurcation theory (cf. \cite{Ni}) then ensures the existence of such nonlinear bound states $\Phi$ bifurcating 
from the zero solution at the eigenvalues $E(t)$. Depending on the type of eigenvalue $E(t)$, these bound states 
are called nonlinear ground states, or nonlinear excited states, respectively (see, e.g., \cite{SoWe2, Ts}). 
To be more precise, we recall the following result:

\begin{lemma} 
\label{lem:bound}
Let $t\in I$ be fixed, $V(t,\cdot) \in \mathcal S(\R^d)$ and denote by $ E(t)<0$ a simple eigenvalue separated from the rest of ${\rm spec}\,( H(t))$. 
For $\lambda >0$ let $E_*(t) \in (E(t),0)$ and for $\lambda <0$, let $E_*(t)<E(t)$. Then there exists a solution $(E_*(t), \Phi(t, \cdot))$ to \eqref{eq:stat},
such that $E_* (t) \mapsto \| \Phi(t,\cdot) \|_{H^2}$ is smooth for $E_*\not = E$ and
\[
\lim_{E_* \to E} \| \Phi(t,\cdot) \|_{H^2} = 0.
\]
In addition, for $\eps_1<1$ sufficiently small and $\frac{1}{\lambda} (E_* - E)< \eps_1$ we have
\[
\left \| \, \Phi (t,\cdot)- \left(\frac{E_*(t)-E(t)}{\mu(t)}\right)^{\frac{1}{2\sigma}} \, \chi(t,\cdot) \,  \right \|_{H^2(\R^d)} = O(E_*(t)-E(t)),
\]
where $\chi(t,\cdot)\in L^2(\R^3)$ is a normalized eigenfunction associated to the linear eigenvalue $E(t)$ and 
\[
\mu(t) := \lambda  \| \chi(t,\cdot) \|^{2\sigma+2}_{L^{2\sigma+2}}.
\] 
%In addition, for any $M \le \eps_1$, there is a unique family of bound states $t\mapsto \Phi(t,\cdot) \in C^1_{\rm b}(I; H^2(\R^3))$ with constant mass $\| \Phi(t,\cdot)\|^2_{L^2}=M$.
\end{lemma}

This result (for time-independent potentials $V=V(x)$) is stated in \cite[Theorem 2.1]{SoWe} in the case where 
the linear Hamiltonian has exactly one eigenvalue $E<0$. The formula for $\Phi$ is a result of a perturbation calculation and already stated in \cite{RoWe}. 
A more detailed proof in the case where $H$ admits exactly two simple eigenvalues, $d=3$, $\sigma =1$, and 
$\lambda >0$, can be found in \cite[Lemma 2.1]{TY}. The situation for several linear eigenvalues is discussed in, e.g., \cite{Ts}, while \cite{GP} studies the the case of degenerate 
eigenvalues.

\begin{remark}
One should note that in several of the aforementioned works (cf. \cite{Ts, TY, SoWe2}), the additional assumption that $0$ is 
neither an eigenvalue nor a resonance for the linear Hamiltonian $H = -\Delta +V$ is imposed. 
This condition, however, is not used in the existence proof of nonlinear bound states (it ensures the applicability of certain dispersive estimates for the associated Schr\"odinger group \cite{PiWa}).
\end{remark}

Hence, for any fixed $t\in I$, there is an $\eps_1(t)>0$, such that as long as $\frac{1}{\lambda} (E_*(t) - E(t))< \eps_1(t)$ we have a nonlinear bound state of size 
\[ 
M(t):=\| \Phi(t,\cdot) \|_{L^2} = O(\eps_1^{1/(2\sigma)}).
\]
Unfortunately, Lemma \ref{lem:bound} does not provide any assertion on the time-dependence of $t\mapsto M(t)$. 
In fact it has been shown in \cite{Gr}, that for a given compact time-interval $ J\subset I$, and $M>0$ sufficiently small there exists a 
unique family of bound states $t\mapsto \Phi (t,\cdot)\in H^2(\R^d))$ continuously depending on time, and with constant mass $\| \Phi(t,\cdot)\|^2_{L^2} = M$. However, since \cite{Gr} is 
currently only published in a PhD thesis (the corresponding paper \cite{GG} is being finalized), we shall not take this result for granted but only 
assume that such a property holds. More precisely, we impose: 

\begin{hyp}\label{ass3}
Let $J\subset I$ be a compact time-interval and assume 
that there is an $M_0(J)>0$ such that for all $M\le M_0$, 
there exists a unique family of solutions $t\mapsto \Phi (t,x)$ to \eqref{eq:stat} with $\Phi  \in C(J; H^2(\R^d))$ and $\| \Phi(t,\cdot)\|^2_{L^2} = M$. 
\end{hyp}

This allows us to prove prove the following result, in which $\eps_0(J)>0$ denotes the same constant as in Theorem \ref{thm:main}.

\begin{proposition}\label{prop:nonlin1}
Let $\sigma \in \N$, $\lambda \in \R$, and Assumptions \ref{ass1}, \ref{ass2} and \ref{ass3} hold. Moreover, let $\eps \in (0,\min(\eps_0(J), M_0(J))]$
and denote by $\Phi^\e(t,x)$ a family of bound states such that $\| \Phi^\e(t,\cdot)\|_{L^2}=\sqrt{\e}$. Finally, let  
$\Psi^\e$ be the solution to \eqref{eq:NLS} with initial data $\Psi^\e_{\rm in}= \sqrt{\e} \psi_{\rm in}^\e$ where $\psi^{\e}_{\rm in}$ satisfies the conditions of Theorem \ref{thm:main}. 
Then there exist a constant $K>0$ such that
\[
\sup_{t\in J} \left \| \Psi^\e(t,\cdot)  -  \Phi^\eps(t,\cdot) e^{-i \varphi^\e(t)} \right\|_{L^2(\R^3)} \le K \eps^{\min(3/2, \sigma)},
\]
where $\varphi^\e$ is given by
\[
\varphi^\e(t) = \frac{1}{\eps} \int_{t_0}^t E(s) ds +  \lambda \int_{t_0}^t \| \chi(s,\cdot)\|_{L^{2\sigma +2}}^{2\sigma+2} \,ds -i  \beta(t) .
\]
\end{proposition}

\begin{proof}
Using the triangle inequality we obtain
\begin{align*}
 \left \| \Psi^\e(t,\cdot)  -  \Phi^\e(t,\cdot) e^{-i \varphi^\e(t)} \right\|_{L^2} \le & \, \left \| \Psi^\e(t,\cdot)  -  \sqrt{\eps} \chi(t,\cdot) e^{-i \varphi^\e(t)} \right\|_{L^2} \\
 & \, +  \left \| \sqrt{\eps} \chi(t,\cdot)  -  \Phi^\e(t,\cdot) \right\|_{L^2}
 \end{align*}
 where in the second term on the right hand side we have used the fact that $\varphi^\e$ is purely time-dependent. In view of Theorem \ref{thm:main}, 
 the first term on the right hand side is $O(\eps^{3/2})$, uniformly for $t\in J$, where $J$ is any compact time interval $J\subset I$. 
 In order to estimate the second term, we use the triangle inequality once more to obtain
\begin{align*}
 \left \| \sqrt{\eps} \chi(t,\cdot)  -  \Phi^\e(t,\cdot) \right\|_{L^2} \le  &\, \left \|\left(\frac{E_*(t)-E(t)}{\mu(t)}\right)^{\frac{1}{2\sigma}}\, \chi(t,\cdot)   -  \Phi^\e(t,\cdot) \right\|_{L^2} \\
 & \,  +  
 \left | \left(\frac{E_*(t)-E(t)}{\mu(t)}\right)^{\frac{1}{2\sigma}} - \sqrt{\eps} \right |,
 \end{align*} 
 since $\| \chi (t,\cdot) \|_{L^2} =1$, by assumption. From Lemma \ref{lem:bound} we find that for all $t\in J$:
 \[
\left | \| \Phi^\e \|_{L^2} - \left(\frac{E_*(t)-E(t)}{\mu(t)}\right)^{\frac{1}{2\sigma}} \right| \le O(E_*(t) - E(t)),
 \]
 which together with the fact that $\| \Phi^\e(t,\cdot) \|_{L^2} =\sqrt{\eps} $, implies
 \begin{equation}
 \label{eq:epsest}
 \eps = \left(\frac{E_*(t)-E(t)}{\mu(t)}\right)^{\frac{1}{\sigma}}+O\big((E_*(t)-E(t))^{1+{\frac{1}{2\sigma}}}\big).
 \end{equation}
 In particular \eqref{eq:epsest} together with Lemma \ref{lem:bound} implies
 \[
 \left \| \left(\frac{E_*(t)-E(t)}{\mu(t)}\right)^{\frac{1}{2\sigma}}\, \chi(t,\cdot)   -  \Phi^\e(t,\cdot) \right\|_{L^2} \lesssim \e^{\sigma}.
 \]
On the other hand, using $\sqrt{a} - \sqrt{b} = (a-b)/ (\sqrt{a} + \sqrt{b})$, we also have 
 \[
 \left |  \left(\frac{E_*(t)-E(t)}{\mu(t)}\right)^{\frac{1}{2\sigma}} - \sqrt{\eps} \right | \lesssim \eps^\sigma.
 \]
 In summary, we find that for $t\in  J$ it holds
 \[
 \left \| \Psi^\e(t,\cdot)  -  \Phi(t,\cdot) e^{-i \varphi^\e(t)} \right\|_{L^2} \lesssim \eps^{3/2} + \eps^\sigma \lesssim \eps^{\min(3/2, \sigma)},
 \]
 which yields the desired result.
\end{proof}

\begin{remark}
Interestingly, the proof shows that in the asymptotic regime considered and in the case $\sigma =1$ (cubic nonlinearity), 
the linear eigenfunction (endowed with the appropriate phase factor) satisfies a better approximation estimate than the nonlinear one.
However, for $\sigma \ge 2$, the approximation rate is the same for both the linear and the nonlinear subspace.  
\end{remark}

\begin{proof}[Proof of Corollary \ref{cor:nonlin}] In order to estimate the operator norm of 
\[
|\Psi^\e \rangle \langle  \Psi^\e  | -  |\Phi^\e \rangle \langle  \Phi^\e | = |\Psi^\e \rangle \langle  \Psi^\e  | -  |\Phi^\e e^{-i \varphi^\e} \rangle \langle  \Phi^\e e^{-i \varphi^\e} |\, , 
\]
it is enough to consider $f\in L^2(\R^d)$, with $\| f\|_{L^2}=1$ and estimate
\begin{align*}
 &  \| |\Psi^\e \rangle \langle  \Psi^\e , f\rangle -   |\Phi^\e e^{-i \varphi^\e} \rangle \langle  \Phi^\e  e^{-i \varphi^\e}  , f\rangle  \|_{L^2} \\
& \le \| |\Psi^\e \rangle \langle  \Psi^\e , f\rangle -  |\Psi^\e \rangle \langle  \Phi^\e e^{-i \varphi^\e} , f\rangle  \|_{L^2}  \\
& \quad  +  \| |\Psi^\e \rangle \langle  \Phi^\e e^{-i \varphi^\e} , f\rangle -  |\Phi^\e e^{-i \varphi^\e} \rangle \langle  \Phi^\e e^{-i \varphi^\e} f\rangle  \|_{L^2}.
\end{align*}
Cauchy-Schwarz implies that the first term on the right hand side can be estimated by
\[
\| |\Psi^\e \rangle \langle  \Psi^\e , f\rangle -  |\Psi^\e \rangle \langle  \Phi^\e e^{-i \varphi^\e}, f\rangle  \|_{L^2} \le \| \Psi^\e \|_{L^2} \| \Psi^\e -  \Phi^\e e^{-i \varphi^\e} \|_{L^2} ,
\]
since $\|f\|_{L^2} = 1$. Analogously we get
\[
\| |\Psi^\e \rangle \langle  \Phi^\e e^{-i \varphi^\e} , f\rangle -  |\Phi^\e e^{-i \varphi^\e} \rangle \langle  \Phi^\e e^{-i \varphi^\e} f\rangle  \|_{L^2} \le  \| \Phi^\e \|_{L^2} \| \Psi^\e -  \Phi^\e e^{-i \varphi^\e} \|_{L^2} .
\]
Having in mind that $\| \Phi^\e\|_{L^2} \simeq \| \Psi^\e \|_{L^2} = O(\sqrt{\eps})$, this, together with the result of Proposition \ref{prop:nonlin1}, then yields
\[
 \big \| |\Psi^\e(t,\cdot)\rangle \langle  \Psi^\e(t,\cdot) | -  |\Phi^\e(t,\cdot)\rangle \langle  \Phi^\e(t,\cdot)| \big\|_{L^2\to L^2}  \lesssim \eps^{\min(2, \sigma +1/2)},
\]
which implies the assertion of Corollary \ref{cor:nonlin}.
\end{proof}

The main drawback of Proposition \ref{prop:nonlin1} is the fact that the phase $\varphi^\e$ involves information of the linear eigenvalue problem. In particular 
the dynamical phase involves the linear eigenvalue $E(t)$ instead of $E_*(t)$. This, however, can be remedied in the case of sufficiently strong nonlinearities.

\begin{corollary}\label{cor:nonlin1}
Under the same assumptions as in Proposition \ref{prop:nonlin1} but for $\sigma \ge 2$, there exists a solution 
$\Psi^\e$ of \eqref{eq:NLS} and a $K>0$ such that
\[
\sup_{t\in J} \left \| \Psi^\e(t,\cdot)  -  \Phi^\eps(t,\cdot) e^{-i \varphi_*^\e(t)} \right\|_{L^2(\R^3)} \le K \eps^{\sigma-1},
\]
where $\varphi_*^\e$ is given by
\[
\varphi_*^\e(t) = \frac{1}{\eps} \int_{t_0}^t E_*(s) ds +  \lambda \int_{t_0}^t \| \chi(s,\cdot)\|_{L^{2\sigma +2}}^{2\sigma+2} \,ds  .
\]
\end{corollary}
In other words, for nonlinearities with powers larger than cubic, we can obtain the physically ``correct" dynamical phase. 
Here, we also implicitly assume that $\chi$ is chosen to be real valued, and thus 
$\beta(t)\equiv 0$, which is always possible as long as the evolution is not cyclic in time (see Remark \ref{rem:berry}).

\begin{proof} We choose a real-valued family of linear eigenfunctions $\chi(t,x)$, which implies $\beta(t)\equiv 0$.
Using the triangle inequality we have
\begin{align*}
 \left \| \Psi^\e(t,\cdot)  -  \Phi^\e(t,\cdot) e^{-i \varphi_*^\e(t)} \right\|_{L^2} \le & \, \left \| \Psi^\e(t,\cdot)  -  \Phi^\e(t,\cdot) e^{-i \varphi^\e(t)} \right\|_{L^2} \\
 & \, +  \left \|\ \Phi^\e(t,\cdot) \left(e^{-i \varphi^\e(t)} - e^{-i \varphi_*^\e(t)}\right)  \right\|_{L^2}.
 \end{align*}
Here the first term on the right hand side is of order $O(\eps^{\min(3/2, \sigma)})$, in view of Proposition \ref{prop:nonlin1}. To estimate the second term, we write 
\begin{align*}
& \left \|\ \Phi^\e(t,\cdot) \left(e^{-i \varphi^\e(t)} - e^{-i \varphi_*^\e(t)}\right)  \right\|_{L^2} = \left \|\ \Phi^\e(t,\cdot) \left(e^{-i( \varphi^\e(t) - \varphi_*^\e(t))} - 1 \right)  \right\|_{L^2}\\
& \le \| \Phi ^\e(t,\cdot) \|_{L^2} \left |e^{-i( \varphi^\e(t) - \varphi_*^\e(t))} - 1 \right |. 
\end{align*}
Recalling the definition of the phases and the fact that $| e^{i \theta}-1| \le |\theta|$, we obtain
\[
 \left |e^{-i( \varphi^\e(t) - \varphi_*^\e(t))} - 1 \right | \le \frac{1}{\eps} \int_{t_0}^t | E(s) - E_*(s) | ds \lesssim \eps^{\sigma-1},
\]
in view of \eqref{eq:epsest}. 
\end{proof}

In certain dimensions (and for certain $\sigma\ge 2$), a more careful analysis, using the result of Lemma \ref{lem:bound} together with a Gagliardo-Nirenberg type inequality, would allow a similar statement in which 
the $\| \chi(s,\cdot)\|_{L^{2\sigma +2}}$ appearing in $\varphi^\e_*$ is replaced by $\| \Phi^\e(s,\cdot)\|_{L^{2\sigma +2}}$. In this case, the phase function consequently only depends on 
information given by the nonlinear eigenvalue problem \eqref{eq:stat}. In view of the many assumptions needed, however, 
this result seems to be very far from optimal and we shall therefore not pursue it any further.

\medskip

{\bf Acknowledgment.} The author wants to thank I. and G. Nenciu for many helpful discussions and the anonymous referees for their detailed comments.

%%%%%%%%%%%%%%%%%%%%%%%%%%%%%%%%%%%%%%%%%%%%%%%%

\bibliographystyle{amsplain}

\end{document}